\documentclass[aps,pra,twocolumn,preprintnumbers,superscriptaddress,nofootinbib]{revtex4-1}
\usepackage{diagbox}
\usepackage{colortbl}
\usepackage{amsmath}
\usepackage{mathrsfs}
\usepackage{amsthm}
\usepackage{braket}
\usepackage{setspace}
\usepackage{titlesec}
\usepackage{amsfonts}
\usepackage{graphicx}
\usepackage{tikz}
\usepackage{pgfplots}
\usetikzlibrary{graphs}
\usepackage{color}
\usepackage{setspace}
\usepackage{bm}
\usepackage[shortlabels]{enumitem}
\usepackage{IEEEtrantools}
\usepackage{bbold}
\usepackage{mathtools}
\usepackage{amssymb}
\usepackage{multirow}
\usepackage{boldline, makecell, booktabs}

\usepackage{changes}

\newtheorem{defn}{Definition}
\newtheorem{lemma}{Lemma}

\newcommand{\Target}{\mathrm{PR}_v(ab|xy)}
\newcommand{\PNrounds}{P_{\boldsymbol{\alpha}}(\mathbf{a},\mathbf{b}|\mathbf{x},\mathbf{y})}

\begin{document}

\author{Boris Bourdoncle}
\affiliation{ICFO-Institut de Ci\`encies Fot\`oniques, The Barcelona Institute of Science and Technology, 08860 Castelldefels (Barcelona), Spain}
\author{Stefano Pironio}
\affiliation{Laboratoire d'Information Quantique, CP 224, Universit\'e libre de Bruxelles (ULB), 1050 Bruxelles, Belgium}
\author{Antonio Ac\'in}
\affiliation{ICFO-Institut de Ci\`encies Fot\`oniques, The Barcelona Institute of Science and Technology, 08860 Castelldefels (Barcelona), Spain}
\affiliation{ICREA-Instituci\'o Catalana de Recerca i Estudis Avan\c cats, Lluis Companys 23, 08010 Barcelona, Spain}

\title{Quantifying the randomness of copies of noisy Popescu-Rohrlich correlations}

\begin{abstract}
In a no-signaling world, the outputs of a nonlocal box cannot be completely predetermined, a feature that is exploited in many quantum information protocols exploiting non-locality, such as device-independent randomness generation and quantum key distribution. This relation between non-locality and randomness can be formally quantified through the min-entropy, a measure of the unpredictability of the outputs that holds conditioned on the knowledge of any adversary that is limited only by the no-signaling principle. This quantity can easily be computed for the noisy Popescu-Rohrlich (PR) box, the paradigmatic example of non-locality.

In this paper, we consider the min-entropy associated to several copies of noisy PR boxes. In the case where $n$ noisy PR-boxes are implemented using $n$ non-communicating pairs of devices, it is known that each PR-box behaves as an independent biased coin: the min-entropy per PR-box is constant with the number of copies. We show that this doesn't hold in more general scenarios where several noisy PR-boxes are implemented from a single pair of devices, either used sequentially $n$ times or producing $n$ outcome bits in a single run. In this case, the min-entropy per PR-box is smaller than the min-entropy of a single PR-box, and it decreases as the number of copies increases.
\end{abstract}

\maketitle

\section{Introduction}
Devices that are non-locally correlated, i.e., which violate Bell inequalities, necessarily produce outcomes that cannot be perfectly determined \cite{Valentini2002273}. This statement is true even according to theories that can deviate from the standard quantum formalism, provided that they satisfy the no-signaling principle according to which local measurements made on a subsystem cannot reveal information about measurements performed on distant subsystems. 

This relation between non-locality, randomness, and no-signaling can be  illustrated through the paradigmatic example of the noisy PR-box. Suppose that Alice has a device where she can input $x\in\{0,1\}$ (a measurement setting) and which outputs $a\in\{0,1\}$ (the measurement outcome). Similarly, Bob has a device where he can input $y\in\{0,1\}$ and which outputs $b\in\{0,1\}$. Let us assume that Alice's and Bob's devices behave according to the joint probabilities
\begin{equation}\label{eq:pr}
\Target=\begin{cases}
3/8+v/{8}&\text{ if } a+b=xy\mod 2\\
1/8-v/{8}&\text{ otherwise,}
\end{cases}
\end{equation}
parameterized by the number $v\in[-1,1]$. The case $v=1$ corresponds to the ideal PR-box \cite{Popescu1994Quantum}, $v=-1$ to uniform white noise, and the intermediate cases to noisy-PR boxes given by a mixture of these two possibilities. The devices violate the Clauser-Horne-Shimony-Holt (CHSH) inequality \cite{Clauser1969Proposed}, hence are non-local, when $v\in\, ]0,1]$. 
They can be realized through measurement on a quantum state when $v\leq \sqrt{2}-1$, with $v=\sqrt{2}-1$ corresponding to Tsirelson-correlations, i.e.,  correlations reaching the maximal quantum violation of the CHSH inequality \cite{Tsirelson1980Quantum}. 

We can quantify how random Alice's outcome $a$ is by considering how predictable it is to some third party, Eve. Eve could hold information allowing her to guess Alice's outcome $a$ with greater probability that what directly follows from the distribution (\ref{eq:pr}). For instance, it could be that this distribution is realized as a mixture of underlying distributions which are individually less random than (\ref{eq:pr}) and that Eve is aware of which one of these underlying distributions is currently realized. More generally, Eve could hold some physical system correlated to Alice's and Bob's devices and performing a measurement on her system could reveal useful information about Alice's outcome. Denoting $z$ Eve's measurement choice and $e$ the corresponding outcome, we can describe this situation through a tripartite distribution $P(abe|xyz)$, whose marginal distribution for Alice and Bob corresponds to the noisy PR-correlations: $\sum_e P(abe|xyz)=\Target$.

It can easily be shown that, no matter what Eve's strategy is, the maximum probability $G_1(v)$ with which she can guess Alice's outcome $a$ is\footnote{Anticipating a notation that we will use later on, the subscript ``1'' in $G_1(v)$ refers to a single copy of the noisy PR-box (\ref{eq:pr}).}
\begin{equation}
G_1(v)=1-\frac{v}{2}\,.
\label{SolutionN1}
\end{equation}
This value holds under the only assumption that Alice, Bob, and Eve's systems satisfy the no-signaling constraints
\begin{align}
P(ab|xyz)&=P(ab|xy),\nonumber\\
P(ae|xyz)&=P(ae|xz),\label{eq:nosig}\\
P(be|xyz)&=P(be|yz),\nonumber
\end{align}
stating that the input of one's party cannot affect the marginal distribution of the two other remote parties. Eq.~(\ref{SolutionN1}) is proven in Appendix~\ref{app:OneRound} and Eve's optimal strategy is sketched in Fig.~\ref{fig:SolN1}. 

\begin{figure}
\includegraphics{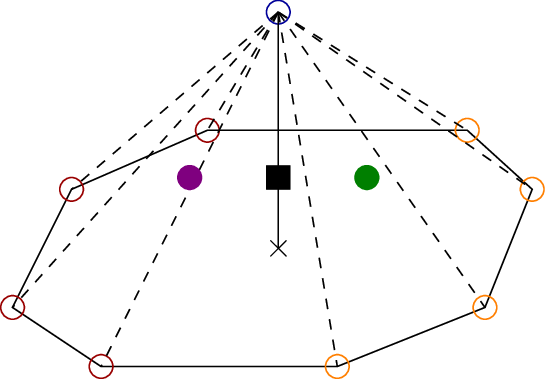}
\caption{Schematic representation of the adversarial strategy that achieves the value given in Eq. \eqref{SolutionN1}. The base of the pyramid represents the CHSH facet of the local set. The eight extreme points on this facet are the eight deterministic strategies attaining CHSH=2. The blue point on top represents the PR-box. For some fixed inputs $x,y$, the local points on the left side (in red) yield the same value for $a$, say 0, and the ones on the right side (in orange) yield the other possible value, say 1. In order to guess the value of $a$, Eve can prepare either a mixture of the red and blue points (in purple), and guess $a=0$, or a mixture of the orange and blue points (in green), and guess $a=1$. On average, these two points reproduce Alice and Bob's expected distributions, $\mathrm{PR}_v$, here depicted by a square.}
\label{fig:SolN1}
\end{figure}

The optimal guessing probability (\ref{SolutionN1}) represents a measure of the randomness of noisy PR-correlations. It is strictly smaller than 1, and thus Alice's outcome cannot be perfectly predicted by Eve, when $v>0$, i.e., when Alice's and Bob's devices are non-local. It is also common to use the min-entropy $H_1(v)=-\log_2 G_1(v)$ to express the randomness of (\ref{SolutionN1}) in bits \cite{Konig2009Operational}. For instance, the ideal PR-correlations have $H_1(1)=1$ bit of randomness, while the Tsirelson-correlations have $H_1(\sqrt{2}-1)=1-\log_2(3-\sqrt{2})\simeq 0.335$ bits of randomness.

This quantitative trade-off between non-locality and randomness can also be determined under the assumption that the entire quantum formalism holds instead of only no-signaling \cite{Colbeck2006Quantum, Pironio2009Device, Pironio2011Random}, and can be established for other type of correlations than the noisy PR ones \cite{Barrett2005No, Barrett2006Maximally, Nieto2014Using}. Such measures of randomness have not only a fundamental interest but also find direct application in device-independent quantum cryptography protocols, such as randomness certification \cite{Pironio2011Random,Colbeck2006Quantum} and key distribution \cite{Barrett2005No, Acin2006From, AcinPRL2007, Masanes2014Full}.

In this work, we investigate the randomness of noisy-PR correlations in a scenario where Alice and Bob make $n$ observations each, instead of a single one. This operationally corresponds to Alice and Bob using $n$ times a single pair of devices, instead of a single one, either because they use $n$ devices or a single device repeatedly $n$ times. They thus end up with, respectively, input strings $\mathbf{x}=(x_1,\ldots,x_n)$ and $\mathbf{y}=(y_1,\ldots,y_n)$ and output strings $\mathbf{a}=(a_1,\ldots,a_n)$ and $\mathbf{b}=(b_1,\ldots,b_n)$. We assume that Alice and Bob's observations are distributed according to 
\begin{equation}\label{eq:prod}
P(\mathbf{a},\mathbf{b}|\mathbf{x},\mathbf{y})=\prod_{i=1}^n \mathrm{PR}_v(a_ib_i|x_iy_i)\,.
\end{equation}
This means that, from Alice and Bob's perspective, their outputs are the same as if they had used $n$ identical and independent copies of the noisy PR-correlations (\ref{eq:pr}). This example was also studied in \cite{Fitzi2010NonLocality}, where the authors investigate its so-called local part. Note that, even though their results have some similarities with ours, there is no direct connection between the local part of some non-local correlation and its unpredictability.

We here ask how predictable Alice's string $\mathbf{a}$ is to some third party, Eve, under the sole assumption of no-signaling. In full generality, we can again characterize correlations among Alice, Bob, and Eve through a $2n+1$-partite distribution $P(\mathbf{a},\mathbf{b},e|\mathbf{x},\mathbf{y},z)$, consisting of $n$ input and output bits for Alice, $n$ input and output bits for Bob, and a single input and output symbol for Eve. 

There are, however, different ways to generalize the no-signaling conditions (\ref{eq:nosig}) to our $2n+1$-partite situation, depending on how Alice and Bob's experiment is performed (see Fig.~\ref{fig:NSconditions}). For instance, Alice and Bob could use $n$ separated pairs of devices, where each pair $i=1,\ldots,n$ receives inputs $x_i,y_i$ and produces outputs $a_i,b_i$. They could use a single pair of devices $n$ times in succession, where now $x_i,y_i$ and $a_i,b_i$ refers to the inputs and outputs at the $i$th round. A further possibility is that Alice holds some big device where she directly inputs $n$-bit strings $\mathbf{x}$ and get $n$-bits output strings $\mathbf{a}$, and similarly Bob holds a big device accepting $n$-bit inputs $\mathbf{y}$ and producing $n$-bit outputs $\mathbf{b}$. To each such physical scenario is associated a different set of no-signaling constraints corresponding to limitations on how the input $x_i$ (or $y_i$) can causally influence the output strings $\mathbf{a}$ and $\mathbf{b}$. In what follows, we will define in more details four natural scenarios and their associated no-signaling constraints. 

In all cases, one possible strategy for Eve is to guess each of Alice's output $a_i$ independently using the optimal single-copy strategy yielding (\ref{SolutionN1}).
However, there may exist clever strategies that performs better than this independent guessing strategy. This is so even though the correlations (\ref{eq:prod}) look identical and independent from Alice's and Bob's perspective, because they need not look that way from Eve's point of view. Indeed, the probabilities $P(\mathbf{a},\mathbf{b}|\mathbf{x},\mathbf{y},e,z)$ conditioned on Eve's knowledge do not need to take a product form, only their average $\sum_e P(e|z) P(\mathbf{a},\mathbf{b}|\mathbf{x},\mathbf{y},e,z)=P(\mathbf{a},\mathbf{b}|\mathbf{x},\mathbf{y})$, corresponding to tracing out Eve, should. In particular, Eve can design the correlations $P(\mathbf{a},\mathbf{b}|\mathbf{x},\mathbf{y},e,z)$ in such a way that the distribution of an output pair $(a_i,b_i)$ is correlated with other values of inputs and outputs. This enables Eve to increase the predictability of some particular sequences, conditioned on the value of $e$, while keeping Alice and Bob's marginal distributions unchanged.

We show that this is indeed what happens for several no-signaling scenarios of interest. The single-copy guessing probability (\ref{eq:guess}) thus does not correctly reflect the randomness of noisy PR-boxes in a situation involving $n$ copies of such correlations. 

Beside its fundamental interest, our work is also motivated by the problem of understanding better the security of quantum key distribution and quantum random number generation against no-signaling adversaries, whose status is not clear at the moment \cite{Arnon2012Limits, Salwey2016Stronger}. Previous works have looked at how much information a no-signaling adversary can obtain about the outcomes of $n$ PR-boxes after privacy amplification \cite{Arnon2012Towards}. We look here at her information before privacy amplification, i.e., on the raw output string. Though the results that we present do not have yet direct implications for the security of quantum key distribution and quantum random number generation schemes, they contribute to a better characterization of adversarial strategies.

Before presenting and discussing our results, we introduce in the next section the problem that we consider in more detail.

\section{Definitions}
\subsection{General scenario}
We use bold variables to denote $n$-bit strings, e.g., $\mathbf{x}=(x_1,\ldots,x_n)\in\{0,1\}^n$. Subscripts are used to denote certain sub-strings of these $n$-bit strings, e.g. $\mathbf{x}_{\leq i}=(x_1,\ldots,x_i)$, $\mathbf{x}_{>i}=(x_{i+1},\ldots,x_n)$ or $\mathbf{x}_{\setminus i}=(x_1,\ldots,x_{i-1},x_{i+1},\ldots,x_n)$.
The subscript 0 corresponds to the empty string: $\mathbf{x}_0=\emptyset$. 

As in the introduction, we consider a situation where Alice, after interacting $n$ times with one or several devices, ends up with input and output strings $\mathbf{x}$ and $\mathbf{a}$. Similarly, Bob ends up with input and output strings $\mathbf{y}$ and $\mathbf{b}$. We assume, as in Eq.~(\ref{eq:prod}), that the joint probabilities $P(\mathbf{a},\mathbf{b}|\mathbf{x},\mathbf{y})$ correspond to $n$-copies of noisy PR-correlations.

We assume that Eve holds a system that may be correlated to Alice's and Bob's devices, a situation that can be described, as in the introduction, through a distribution $P(\mathbf{a},\mathbf{b},e|\mathbf{x},\mathbf{y},z)$ that is compatible with Alice and Bob marginals. Under the assumption that these correlations cannot be used for signaling between Eve and Alice-Bob, we can describe things in an alternative, convenient way that does not directly involves Eve's input $z$. Indeed, as explained in \cite{Acin2016Necessary}, any measurement that Eve can perform on her system can be interpreted as a choice of a convex decomposition
\begin{equation}\label{eq:ABdec}
\sum_e P(e) P_e(\mathbf{a},\mathbf{b}|\mathbf{x},\mathbf{y})
=\prod_{i=1}^n \mathrm{PR}_v(a_ib_i|x_iy_i)
\end{equation}
of Alice's and Bob's devices and her measurement outcome $e$ can be interpreted as indicating one part of this decomposition. Conversely, any convex decomposition (\ref{eq:ABdec}) of Alice and Bob's system can be realized by Eve by choosing an appropriate measurement on her system. From now on, we adopt this view.

The components $P_e(\mathbf{a},\mathbf{b}|\mathbf{x},\mathbf{y})$ in the above decomposition are not arbitrary but should satisfy certain no-signaling constraints reflecting the causal relations that follow from the way Alice and Bob use their devices. We consider four types of such no-signaling constraints.

\begin{defn}[Full-NS] 
The probabilities $P_e(\mathbf{a},\mathbf{b}|\mathbf{x},\mathbf{y})$ are fully no-signaling (Full-NS) if, for every $1\leq i\leq n$,
\begin{align}
P_{e}(\mathbf{a}_{\setminus i},\mathbf{b}|\mathbf{x},\mathbf{y})&=P_{e}(\mathbf{a}_{\setminus i},\mathbf{b}|\mathbf{x}_{\setminus i},\mathbf{y})\label{fullns1}\\
P_{e}(\mathbf{a},\mathbf{b}_{\setminus i}|\mathbf{x},\mathbf{y})&=P_{e}(\mathbf{a},\mathbf{b}_{\setminus i}|\mathbf{x},\mathbf{y}_{\setminus i})\label{fullns2}\,.
\end{align}
\label{def:FNS}
\end{defn}
In the above definition, it is to be understood that Eq.(\ref{fullns1}) holds for all possible values of $\mathbf{a}_{\setminus i},\mathbf{b},\mathbf{x},\mathbf{y},e$ and Eq.(\ref{fullns2}) for all possible values of $\mathbf{a},\mathbf{b}_{\setminus i},\mathbf{x},\mathbf{y},e$. The marginal distribution $P_{e}(\mathbf{a}_{\setminus i},\mathbf{b}|\mathbf{x},\mathbf{y})$ is obtained by summing the whole probability table of Alice and Bob over the missing variables: $P_{e}(\mathbf{a}_{\setminus i},\mathbf{b}|\mathbf{x},\mathbf{y})=\sum_{\mathbf{a}_i}P_{e}(\mathbf{a},\mathbf{b}|\mathbf{x},\mathbf{y})$ and analogously for $P_{e}(\mathbf{a},\mathbf{b}_{\setminus i}|\mathbf{x},\mathbf{y})$. The other definitions that we introduce below should be understood similarly.

This condition corresponds to having $2n$ parties which satisfy all possible pairwise no-signaling conditions. It is operationally equivalent to using $2n$ boxes that are all causally independent, i.e., no communication is allowed between any of them, even though they can be correlated \cite{Masanes2011Secure, Hanggi2009Efficient}. See Fig.~\ref{fig:NSconditions} (a) for a schematic representation of this scenario.

\begin{defn}[ABNS] 
The probabilities $P_e(\mathbf{a},\mathbf{b}|\mathbf{x},\mathbf{y})$ are Alice-Bob no-signaling (ABNS) if
\begin{align}
P_e(\mathbf{b}|\mathbf{x},\mathbf{y})&= P_e(\mathbf{b}|\mathbf{y})\\
P_e(\mathbf{a}|\mathbf{x},\mathbf{y})&=P_e(\mathbf{a}|\mathbf{x})\,.
\end{align}
\label{def:ABNS}
\end{defn}

In this case, no-signaling holds only between Alice and Bob, i.e., there is no communication between them. It means that the inputs used by Bob cannot be inferred from Alice's marginal distribution, even if the information from all the rounds is grouped together, and vice-versa. However, there is no constraint on the internal structure of Alice's or Bob's own marginal. For instance, output $a_1$ could depend on the values of all the inputs $\mathbf{x}=(x_1\ldots x_n)$. 

It is equivalent to considering one big device on Alice's side (respectively Bob's side), that receives as input the string $\mathbf{x}=(x_1\ldots x_n)$ (resp. $\mathbf{y}$) and produces at once the output string $\mathbf{a}=(a_1\ldots a_n)$ (resp. $\mathbf{b}$), or $n$ devices on each side that are used in parallel and can communicate freely amongst themselves \cite{Hanggi2013Impossibility}. This condition is schematically depicted in Fig.~\ref{fig:NSconditions} (b).

\begin{defn}[TONS] 
The probabilities $P_e(\mathbf{a},\mathbf{b}|\mathbf{x},\mathbf{y})$ are time-ordered-no-signaling (TONS) if, for every $0\leq i < n$
\begin{align}
\label{eq:TONSAlice} P_e(\mathbf{a}_{\leq i},\mathbf{b}|\mathbf{x},\mathbf{y})&=P_e(\mathbf{a}_{\leq i},\mathbf{b}|\mathbf{x}_{\leq i},\mathbf{y})\\
\label{eq:TONSBob} P_e(\mathbf{a},\mathbf{b}_{\leq i}|\mathbf{x},\mathbf{y})&=P_e(\mathbf{a},\mathbf{b}_{\leq i}|\mathbf{x},\mathbf{y}_{\leq i})\,.
\end{align}
\label{def:TONS}
\end{defn}
In this case, no-signaling holds between Alice and Bob as for ABNS (take $i=0$). In addition, future rounds (which corresponds to values greater than $i$) have no influence on past rounds (which corresponds to values smaller than $i$) on each side. It describes the situation where two devices are separated from each other during the entire run of the experiment and are used sequentially, while keeping a memory of the past events \cite{Arnon2012Limits, Salwey2016Stronger}. The schematic representation of this condition can be found in Fig.~\ref{fig:NSconditions} (c).
 
Note that Full-NS $\subset$ TONS $\subset$ ABNS.

\begin{defn}[WTONS] 
The probabilities $P_e(\mathbf{a},\mathbf{b}|\mathbf{x},\mathbf{y})$ are weakly time-ordered no-signaling (WTONS) if for all $0\leq i < n$,
\begin{align}
P_e(\mathbf{a}_{\leq i},\mathbf{b}_{\leq i+1}|\mathbf{x},\mathbf{y})&=P_e(\mathbf{a}_{\leq i},\mathbf{b}_{\leq i+1}|\mathbf{x}_{\leq i},\mathbf{y}_{\leq i+1})\\
P_e(\mathbf{a}_{\leq i+1},\mathbf{b}_{\leq i}|\mathbf{x},\mathbf{y})&=P_e(\mathbf{a}_{\leq i+1},\mathbf{b}_{\leq i}|\mathbf{x}_{\leq i+1},\mathbf{y}_{\leq i})\,.
\end{align}
\label{def:WTONS}
\end{defn}
This condition is a weakened version of the time-ordered-no-signaling condition, i.e. TONS $\subset$ WTONS. Future rounds cannot influence past round, and no-signaling holds at each individual round, but, contrarily to ABNS and TONS, no-signaling between Alice and Bob does not hold throughout the entire run of the experiment. It means that Alice's marginal at round $i$ is independent of $\mathbf{x}_{>i}$ and $\mathbf{y}_{\geq i}$, but can depend on $\mathbf{x}_{\leq i}$ and $\mathbf{y}_{< i}$, and likewise for Bob. It describes the situation where two devices are used sequentially and have memory, and where these two devices can moreover communicate between successive rounds \cite{Arnon2012Towards}. See Fig.~\ref{fig:NSconditions} (d) for a schematic representation.

The TONS condition naturally emerges if the two devices can be shielded from each other during the entire experiment, e.g., if $n$ pairs of entangled particles are stored in memory. Yet in many practical situations, pairs of entangled particles are produced one round after the other and distributed to each device. This requires that the devices be opened between each round, at which point some communication between the two devices could happen. WTONS characterizes this situation.

Note that if we consider, as in the WTONS scenario, that communication between the boxes cannot be prevented between the successive rounds, one could also argue that one could not prevent the outcome bits from directly leaking to Eve, thus rendering the notion of guessing probability irrelevant. This point is pertinent in the case of protocols such as Device-Independent Quantum Key Distribution (DIQKD), where Alice and Bob are indeed two distant agents aiming to share some private bits at distant locations. In this case, there is indeed no reason to believe that the information flowing from Alice to Bob could not also flow from Alice to Eve. However, for protocols such as Device-Independent Random Number Generation (DIRNG), Alice and Bob can be thought of as two fictional agents in a single laboratory, as the goal is here to obtain private bits in a unique location. In this situation, we believe that the TONS and WTONS scenarios are two relevant models.     

\begin{figure}
\includegraphics{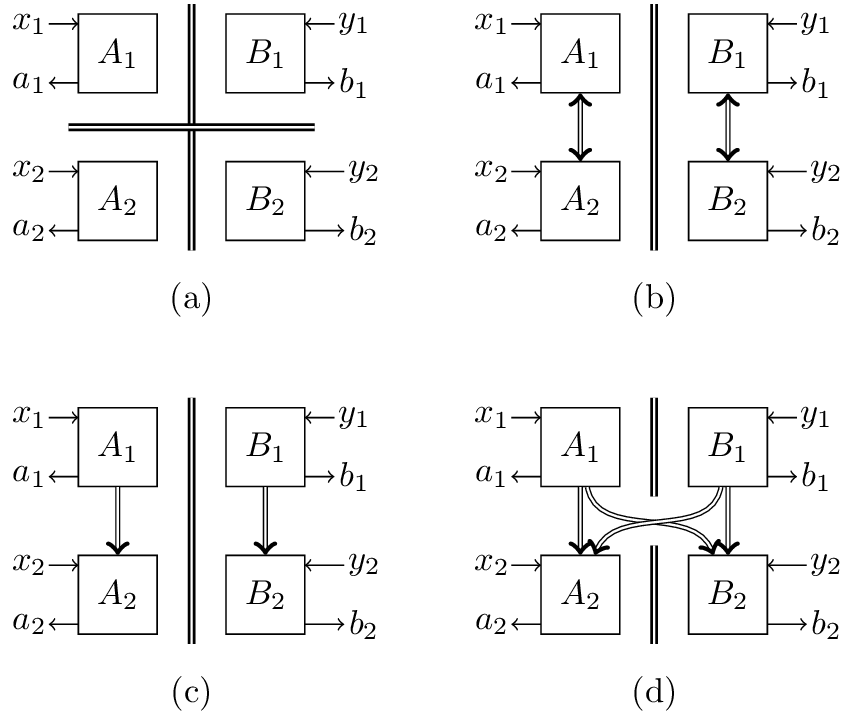}
\caption{Representations of the Full-NS (a), ABNS (b), TONS (c) and WTONS (d) conditions. The double lines indicate that there is no information flow between the corresponding boxes, while the double arrows indicate the direction of the information flow.}
\label{fig:NSconditions}
\end{figure}

\subsection{Quantifying randomness}
Let us now define how we quantify the randomness of Alice's output string $\mathbf{a}$. We do this by considering to which extent Eve can correctly guess these outcomes. As stated previously, we can assume in full generality that from Eve's point of view Alice and Bob's correlations are decomposed as a convex sum (\ref{eq:ABdec}), where with probability $P(e)$ Eve holds the variable $e$ indicating that Alice and Bob's outputs are distributed as $P_e(\mathbf{a},\mathbf{b}|\mathbf{x},\mathbf{y})$. Given $e$, there is then some optimal guess $e\rightarrow \boldsymbol{\alpha}\in\{0,1\}^n$ that Eve can make about Alice's outcome $\mathbf{a}$. We can of course write the convex decomposition (\ref{eq:ABdec}) of Alice and Bob's correlations directly in term of Eve's guess $\boldsymbol{\alpha}$, and thus effectively assume that the arbitrary variable $e$ is already equal to the optimal guess $\boldsymbol{\alpha}$. From now on, we thus assume that $e=\boldsymbol{\alpha}$.

For any $\boldsymbol{\alpha}$, and given that Alice and Bob make the specific choice of inputs $\mathbf{x}^*$ and $\mathbf{y}^*$, the probability that Eve correctly guesses Alice's output is simply the probability $P_{\boldsymbol{\alpha}}(\mathbf{a}=\boldsymbol{\alpha}|\mathbf{x}^*,\mathbf{y}^*)$ that Alice's output $\mathbf{a}$ is equal to Eve's guess $\boldsymbol{\alpha}$. Eve's average guessing probability is thus
\begin{equation}\label{eq:guess}
\sum_{\boldsymbol{\alpha}} P(\boldsymbol{\alpha})P_{\boldsymbol{\alpha}}(\mathbf{a}=\boldsymbol{\alpha}|\mathbf{x}^*,\mathbf{y}^*)\,.
\end{equation}
This average guessing probability holds assuming a given decomposition of Alice and Bob's correlations. But as stated previously, Eve can freely use any convex decomposition of her liking through an appropriate choice of measurement on her system. Eve's maximal guessing probability $G_n(v)$ is thus obtained by maximizing (\ref{eq:guess}) over all possible convex decompositions that are compatible with Alice and Bob's observations and with the no-signaling scenario corresponding to their use of their devices. Eve's optimal guessing probability is thus the solution of the following optimization problem
\begin{align}\label{opt-P-general}
G_n(v) =  & \max   \sum_{\boldsymbol{\alpha}} P(\boldsymbol{\alpha})P_{\boldsymbol{\alpha}}(\mathbf{a}=\boldsymbol{\alpha}|\mathbf{x}^*,\mathbf{y}^*)\\
& \mathrm{s.t.}\phantom{x}  \sum_{\boldsymbol{\alpha}} P(\boldsymbol{\alpha})\PNrounds=\prod_{i=1}^n \mathrm{PR}_v(a_i,b_i|x_i,y_i)\nonumber\\
& \phantom{max} \forall \alpha, \ \PNrounds \text{ is NS}\nonumber
\end{align}
where NS denotes one of the no-signaling constraints NS=\{Full-NS, ABNS, TONS, WTONS\}, depending on which scenario is considered.

It is implicit in the above formulation that Eve's choice of convex decomposition -- and thus that the optimal guessing probability -- depends on the inputs $\mathbf{x}^*$ and $\mathbf{y}^*$ that are chosen by Alice and Bob. We therefore assume that the specific inputs $\mathbf{x}^*$ and $\mathbf{y}^*$ used by Alice and Bob are communicated to Eve. Indeed, our aim is to quantify the fundamental, intrinsic randomness generated at Alice's side, even in a situation where all details of the experimental set-up are known to Eve. From an applied point of view, it also means that this quantity is relevant for a protocol where some actions are taken based on some specific values of inputs $(x^*,y^*)$, fixed in advance: the bound on the predictability is valid even if the protocol is known to Eve.

The optimal guessing probability $G_n(v)$ may therefore depend on the input choices $\mathbf{x}^*$ and $\mathbf{y}^*$ and there could thus be different possible ways to quantify the randomness of Alice's output: e.g., by considering the worst-case over all inputs choices or the expected guessing probability with respect to some  probability distribution for Alice and Bob's inputs. In our case, however, thanks to the symmetries of the noisy PR-correlations (\ref{eq:pr}), the same optimal value $G_n(v)$ is obtained for any possible choice of inputs $\mathbf{x}^*$ and $\mathbf{y}^*$. Indeed, as we show in Appendix~\ref{app:Symmetries}, given any solution to~\eqref{opt-P-general} for a given pair of inputs $\mathbf{x}^*,\mathbf{y}^*$, one can construct a corresponding solution for any other pairs of inputs that yields the same guessing probability. Thus we can simply quantify the randomness of Alice's output through the guessing probability associated to any given input choices. For specificity, we will use the choice $\mathbf{x}^*=\mathbf{y}^*=\mathbf{0}=(0_1,\ldots,0_n)$ in the following.

Note that, even if in our problem the optimal guessing probability $G_n(v)$ is the same for any input choices of Alice and Bob, the particular convex decomposition achieving this optimal value will vary with the choice of inputs. As we said, Eve can remotely choose the optimal decomposition by selecting a measurement on her system when she is informed about Alice and Bob's input choices, and thus $G_n(v)$ correctly reflects the probability with which she can guess Alice's output in the most general scenario. However, this requires Eve to hold some ``coherent memory'', and to delay her measurement until when she is informed about Alice's and Bob's inputs. One could also consider, as in \cite{PironioPRX2013}, a situation where Eve has no such ``coherent memory'' and is forced to commit to a decomposition before Alice's and Bob's inputs are known. Here we choose to quantify randomness in the former scenario because it corresponds to the worst possible setting where Eve's knowledge is maximal. Furthermore, it also corresponds to the scenario where the security of RNG and QKD against no-signaling adversaries is not clearly established.

Finally, note that in the case of the Full-NS, ABNS, and TONS constraints, $P_{\boldsymbol{\alpha}}(\mathbf{a}|\mathbf{x}^*,\mathbf{y}^*)=P_{\boldsymbol{\alpha}}(\mathbf{a}|\mathbf{x}^*)$ and thus Eve's strategy does not actually need to depend on the knowledge of Bob's input $\mathbf{y}^*$. This, however, is not necessarily the case for the WTONS constraints for which no-signaling does not hold between Alice and Bob. This is why we include explicitly $\mathbf{y}^*$ in (\ref{opt-P-general}).

\section{Basic observations and known results}  
Before presenting our actual results -- the optimal solutions to (\ref{opt-P-general}) for different values of $n$, noise levels $v$, and different no-signaling conditions -- let us make some basic observations.

\subsection{Bounds on $G_n(v)$ from $G_1(v)$}

For $n=1$, all the no-signaling conditions NS=\{Full-NS, ABNS, TONS, WTONS\} that we have introduced reduce to the usual no-signaling conditions between Alice and Bob:
\begin{align}
P_e(a_1|x_1,y_1)&=P_e(a_1|x_1)\\
P_e(b_1|x_1,y_1)&=P_e(b_1|y_1)\,.
\end{align}
As we have claimed in the introduction, the optimal guessing probability $G_1(v)$ is known in this case and is given by Eq.~(\ref{SolutionN1}). 

Before attempting to find the guessing probabilities $G_n(v)$ for values of $n>1$, we can already observe that they necessarily satisfy the trivial bounds
\begin{equation}
G_1^n(v)\leq G_n(v)\leq G_1(v)\,
\end{equation}
or explicitly
\begin{equation}\label{eq:bounds}
\left(1-\frac{v}{2}\right)^n\leq G_n(v)\leq 1-\frac{v}{2}\,.
\end{equation}
The lower-bound $G^n(v)\geq G^n_1(v)$ follows from the fact that a possible strategy is for Eve to guess each output bit of Alice $a_i$ independently using the optimal strategy for a single copy of PR-correlations. The probability to guess correctly the entire string $\mathbf{a}=(a_1,\ldots,a_n)$ is then simply the product of the probability to guess correctly each bit independently. There could be, however, more clever strategies, hence this only represents a lower-bound on the $n$-copy guessing probability $G_n(v)$. 

The upper-bound $G_n(v)\leq G_1(v)$ follows from the fact that the probability to guess correctly the entire $n$-bit string $\mathbf{a}$ should not be higher than the probability to guess only one of the $a_i$.

For $v=0$, corresponding to the point at which the noisy PR-correlations become local, the lower-bound and upper-bound coincide and give the trivial value $G_n(0)=1$, as expected since any local correlations admit a purely deterministic explanation.

For $v=1$, corresponding to perfect PR-correlations, it is possible to show that the lower-bound is saturated, i.e., $G_n(1)=(1/2)^n$. This follows from the fact that the product of $n$ perfect PR-correlations is a vertex of the polytopes associated with any of the no-signaling constraints NS=\{Full-NS, ABNS, TONS, WTONS\}, see Appendix \ref{app:ProductPR}.

The values $G_n(v)$ for the different no-signaling constraints that we consider here thus all coincide at the extremities of the interval $v\in[0,1]$ and our problem is to understand how the guessing probability varies as a function of $n$ for $0<v<1$. 

\subsection{$G_n(v)$ in the Full-NS scenario}

For the Full-NS scenario, it happens that the independent strategy discussed above is actually the optimal strategy. This directly follows from the results of Appendix A of \cite{Masanes2014Full}, where it is shown that for every $P(\mathbf{a},\mathbf{b}|\mathbf{x},\mathbf{y})$ that is Full-NS, the following bound holds
\begin{equation}\label{eq:mas}
G_n(v)\leq \sum_{\mathbf{a},\mathbf{b},\mathbf{x},\mathbf{y}} \prod_{i=1}^n \beta(a_i,b_i,x_i,y_i) P(\mathbf{a},\mathbf{b}|\mathbf{x},\mathbf{y})\,,
\end{equation}
where the coefficients $\beta$ are defined as
\begin{equation}\label{eq:beta}
\beta(a,b,x,y)=\begin{cases}
1/8&\text{ if } a+b=xy\mod 2\\
5/8&\text{ otherwise.}
\end{cases}
\end{equation}
In the case where $P(\mathbf{a},\mathbf{b}|\mathbf{x},\mathbf{y})=\prod_{i=1}^n\mathrm{PR}(a_i,b_i|x_i,y_i)$, it is easily seen that this yields $G_n(v)\leq \left(1-\frac{v}{2}\right)^n$. Since this value can be trivially attained with the independent strategy discussed above, we have that 
\begin{equation}\label{eq:istrat}
G_n(v)=\left(1-\frac{v}{2}\right)^n\,.
\end{equation}
The min-entropy
\begin{equation*}
H_n(v)=-\log_2 G_n(v)=-n\log_2\left(1-\frac{v}{2}\right)=nH_1(v)
\end{equation*}
thus scales linearly with $n$: each new use of the noisy-PR correlations brings $H_1(v)$ new bits of randomness.
Interestingly, we show below that this is no longer the case in the other no-signaling scenarios that we consider.

\section{Results}

\begin{figure}
\includegraphics{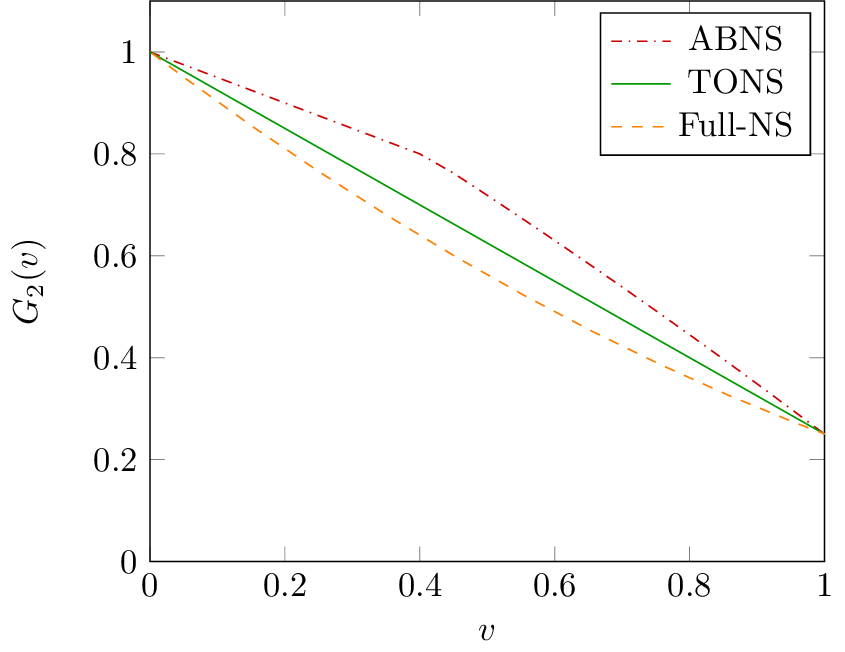}
\caption{Guessing probabilities for $n=2$.}
\label{fig:N2}
\end{figure}

The optimization problem (\ref{opt-P-general}) is a linear program. This is easily seen by rewriting it, as in \cite{Nieto2014Using,Bancal2014More}, in term of the unnormalized probabilities $\tilde P_{\boldsymbol{\alpha}}(\mathbf{a},\mathbf{b}|\mathbf{x},\mathbf{y})=P({\boldsymbol{\alpha}})\PNrounds$. For $n=\{2,3,4,5\}$, we numerically solved this linear program for the three sets ABNS, TONS, WTONS.
 
We find in each case that the optimal guessing probability is higher than the one obtained with the independent strategy corresponding to the lower-bound in (\ref{eq:bounds}).

Furthermore, for the cases $n=\{2,3\}$, we solve (\ref{opt-P-general}) by finding explicit solutions to its primal and dual forms, and thus obtain the analytical expressions of $G_n(v)$. 
\subsection{$n=\{2,3\}$}

The analytical solutions to the optimization problem (\ref{opt-P-general}) are given in Table~\ref{tab:1}, and are plotted as a function of $v$ in Figures~\ref{fig:N2} and \ref{fig:N3}. For the Full-NS scenario, we recover, as expected, the value (\ref{eq:istrat}) corresponding to the independent strategy. In the three other cases ABNS, TONS, WTONS, we find that the guessing probability is strictly higher than this value for all $0<v<1$.

\begin{figure}
\includegraphics{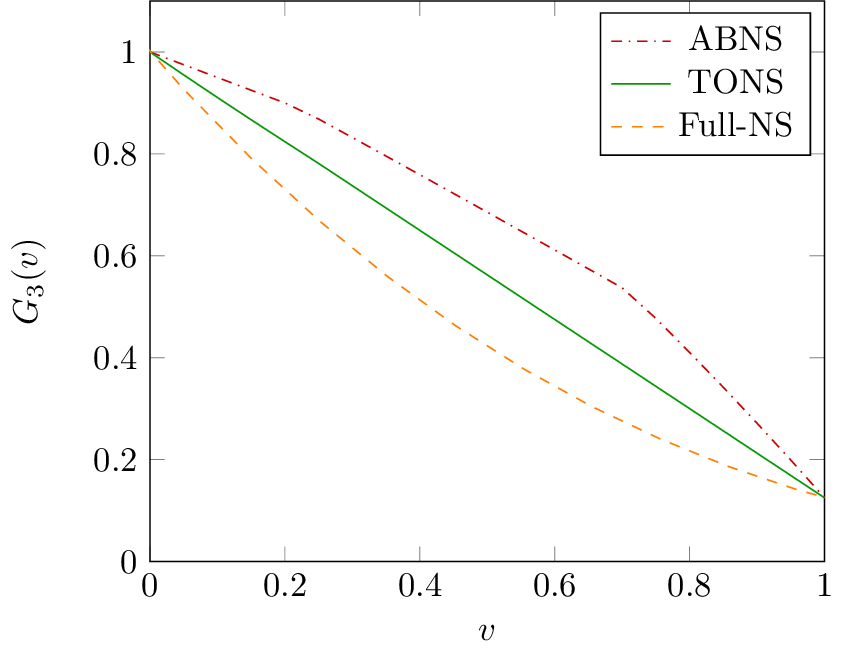}
\caption{Guessing probabilities for $n=3$. The change of behaviour at $v= -2+\sqrt{5}$ indicated in Table~\ref{tab:1} for the TONS and WTONS scenarios is not apparent because the polynomial is very close to the line in this region.}
\label{fig:N3}
\end{figure}

The solutions are detailed in~\cite{Bourdoncle2018Quantifying}. We now make a few observations. First of all, for $n=2$ and $v\leq \sqrt{2}-1$, the guessing probability $G_2(v)$ in the ABNS scenario saturates the trivial upper-bound in (\ref{eq:bounds}) given by the single-round guessing probability, i.e., $G_2(v)= G_1(v)=1-v/2$. This establishes, independently of our dual solutions, that our explicit strategy is optimal in this case.

\begin{table}[t]
\begin{tabular}{l}
\hline\\
 $n=2$, ABNS:\\[6pt]
\qquad$G_2(v) = 
\begin{cases}
1-\frac{1}{2}v&\text{ if } v \leq \sqrt{2}-1 \\
\frac{9}{8}- \frac{3}{4} v-\frac{1}{8} v^2 &\text{ if } v \geq \sqrt{2}-1
\end{cases}$ \vspace{1em}\\
\hline \\
 $n=2$, TONS and WTONS:\\[6pt]
\qquad$G_2(v)=1-\frac{3}{4}v$
\vspace{1em}\\
\hline \\
 $n=3$, ABNS:\\[6pt]
\qquad$G_3(v) = 
\begin{cases}
1-\frac{1}{2} v &\text{ if } v \leq v_1 \\
\frac{67}{64}- \frac{45}{64} v-\frac{3}{64} v^2+\frac{1}{64}v^3 &\text{ if } v_1 \leq v \leq v_2 \\
\frac{41}{32}- \frac{27}{32} v-\frac{9}{32}  v^2-\frac{1}{32}v^3 &\text{ if } v \geq v_2
\end{cases}
$\\
\qquad\parbox{0.9\linewidth}{where $v_1$ is the unique root of $x^3-3x^2-13x+3$ in $[0,1]$
 $(v_1 \approx 0.22038)$ and $v_2$ the unique root of $x^3+5x^2+3x-5$ in $[0,1]$ $(v_2 \approx 0.70928)$.}
\vspace{1em}\\
\hline \\ 
 $n=3$, TONS and WTONS:\\[6pt]
\qquad$G_3(v) = 
\begin{cases}
1-\frac{29}{32} v+\frac{1}{8}v^2+\frac{1}{32}v^3 &\text{ if } v \leq \sqrt{5}-2 \\
1-\frac{7}{8} v &\text{ if } v \geq \sqrt{5}-2
\end{cases}$\vspace{1em}\\
\hline
\end{tabular}
\caption{Guessing probabilities for $n=2,3$.}
\label{tab:1}
\end{table}

Conceptually, it is surprising that the guessing probability does not decrease from $n=1$ to $n=2$ as it means that it is not more difficult for Eve to guess two outcome bits of Alice than it is to guess a single one. More surprisingly, the region where this happens corresponds to $v\leq \sqrt{2}-1$, i.e., to the region where the noisy PR-correlations admit a quantum representation. We do not know whether this is merely a computational coincidence or whether it has some deeper meaning about the structure of the quantum set.

For $n=3$, there is again a region, corresponding to $v\leq 0.22038$, where $G_3(v)=G_1(v)$. This region is smaller than the previous one, but on the other hand, Eve can now guess three successive bits of Alice with the same error probability as when guessing a single one. 

For the TONS and WTONS scenarios, we find that the two solutions coincide. Interestingly, we find that the optimal solution in the case $n=2$ is linear in $v$, as for $n=1$. For $n=3$, this is only true if $v$ is above the threshold $v\geq \sqrt{5}-2$. We now intuitively explain how the strategies we have found work and the origin of this linear behavior.

In our model, Eve distributes the correlations for Alice and Bob and can adapt the decomposition for each round depending on what happened in the previous rounds. For the first round, there is no past, so she prepares the mixture of extremal local and non-local points compatible with Alice and Bob's probabilities depicted in Figure~\ref{fig:SolN1}.

The distribution for the second round depends on what happened in the first one \footnote{Let us stress that Eve doesn't need to acquire this knowledge for the strategy to be valid. This is merely a way to give an intuition about the strategy by decomposing it sequentially, while the attack is entirely designed prior to the experiment.}. If Alice's first output is such that Eve's guess is correct, the devices on the second round behave in a more predictable way, i.e., their correlations correspond to a more local point. This allows Eve to improve her guess on the two generated outputs. On the other hand, if Alice's first output is such that Eve's guess is wrong, the subsequent events are of no importance to the value of the guessing probability: the devices can be maximally non-local, i.e., a PR-box.

These different possibilities can then combine in such a way that Alice and Bob's marginal distributions are as expected, if Eve accurately adjusts the amount of non-locality in the second round based on the value of $v$. For $n=2$, the balance is such that the guessing probability is linear in $v$. 

One could hope to straightforwardly extend this strategy to any number of rounds and that it would imply that the guessing probability be equal to $1-(2^n-1)/(2^n)\cdot v$ for all $n$. This is however not the case. To understand why, note that, in order to constantly improve her guess, Eve needs to prepare distributions that have more and more predictable outcomes, i.e., points that are closer and closer to the local set. But when a point is local, its outcomes are perfectly known to Eve: its predictability cannot increase anymore. We observe that, when this happens at some round, $G_n(v)$ is less than $1-(2^n-1)/(2^n)\cdot v$ for subsequent rounds. This phenomenon happens after a certain number of rounds, which depends on the value of $v$. For $n=3$, we observe it for $v \leq -2+\sqrt{5}$.

\subsection{$n=\{4,5\}$}

\begin{figure}
\includegraphics{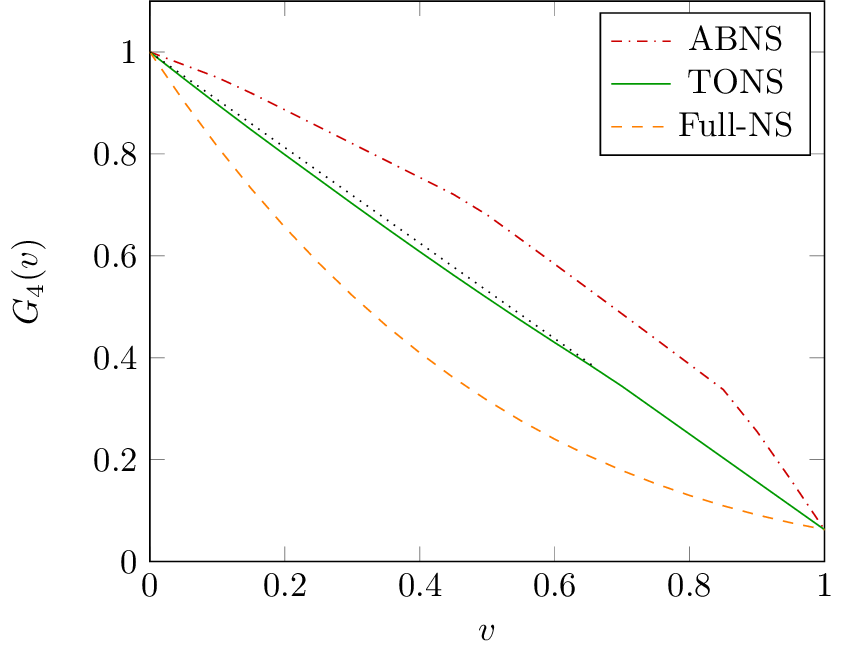}
\caption{Guessing probabilities for $n=4$. We add the line interpolating $(0,1)$ and $(1,\frac{1}{16})$ (black dotted line) to emphasize the breakdown of linear dependence of the TONS guessing probability for some $v$.}
\label{fig:N4}
\end{figure}

We then numerically solved (\ref{opt-P-general}) for $n=\{4,5\}$. The results are plotted in Figures~\ref{fig:N4} and \ref{fig:N5}. In this case, we did not attempt to find the analytical expressions of $G_4(v)$ and $G_5(v)$: keeping track of the dual's variables, which grow exponentially with $n$, becomes demanding, while a numerical result is sufficient for our purpose. 

As before, we observe that for $v$ small enough, there is a region, that gets smaller as $n$ increases, where $G_5(v)=G_4(v)=G_1(v)$ in the ABNS scenario.

For the TONS and WTONS scenarios, the guessing probability depends linearly on $v$ (as $1-15/16\cdot v$ for $n=4$ and as $1-31/32 \cdot v$ for $n=5$ ) when $v$ is large enough. The minimal $v$ for which this happens increases and gets closer to 1 as $n$ increases. 

However, while for $n\leq 3$, the guessing probability is the same for the TONS and WTONS scenarios, this is no longer the case when $n \geq 4$, except in the linear regime for $v$ close to 1. The difference between the TONS and WTONS values is not visible on the graphs, which is why we highlight it in the following tables: 

$$
\begin{array}{|c|c|c|c|c|c|}
\hline 
& v & 0.05 & 0.1 & 0.15 & 0.2   \\
\hline 
\multirow{2}{*}{$G_4(v)$} & WTONS & 0.9487 & 0.8981 & 0.8482 & 0.7990   \\
\cline{2-6}
& TONS & 0.9481 & 0.8972 & 0.8473 & 0.7985   \\
\hline
\end{array}
$$

$$
\begin{array}{|c|c|c|c|c|c|}
\hline 
& v & 0.05 & 0.1 & 0.15 & 0.2   \\
\hline 
\multirow{2}{*}{$G_5(v)$} & WTONS & 0.9451 & 0.8913 & 0.8387 & 0.7865   \\
\cline{2-6}
& TONS & 0.9431 & 0.8874 & 0.8328 & 0.7795   \\
\hline
\end{array}
$$

\begin{figure}
\includegraphics{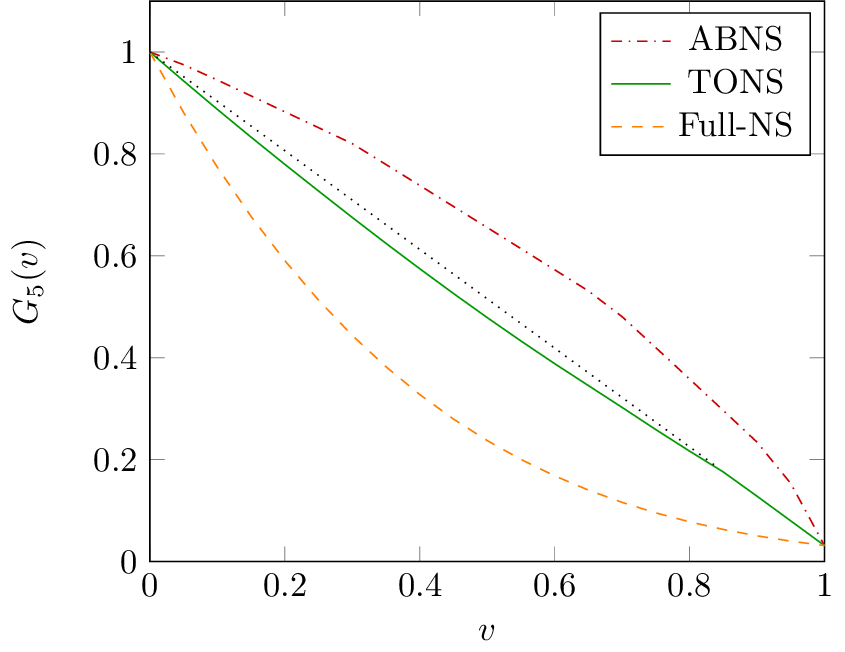}
\caption{Guessing probabilities for $n=5$. We add the line interpolating $(0,1)$ and $(1,\frac{1}{32})$ (black dotted line) to emphasize the breakdown of linear dependence of the TONS guessing probability for some $v$.}
\label{fig:N5}
\end{figure}

Carrying out the numerical optimization for larger $n$ becomes computationally too demanding, as the number of variables and constraints grows exponentially with $n$.  However, the results obtained for small $n$ already have implications for all $n$, as explained below.

\subsection{Implications for all $n$}
For the ABNS, TONS, and WTONS scenarios, we have found in the previous subsections that, contrarily to what happens in the Full-NS scenario, the independent strategy is not the optimal strategy for $n=\{2,3,4,5\}$, i.e., $G_n(v)>G^n_1(v)$. 

This implies in particular that one can improve the lower-bound $G_n(v)\geq G^n_1(v)$ for all $n$, as, instead of considering strategies where Eve guesses independently each individual outcome bit of Alice, one can now consider strategies where Eve guesses independently pairs, triples, quadruples or quintuples of outcome bits of Alice. For instance if $n=5k$, Eve can guess every successive quintuple of outcomes independently, and we have thus the lower-bound
\begin{equation}
G_n(v)=G_{5k}(v)\geq G_5^k(v)\,.
\end{equation}
In term of the min-entropy per run this corresponds to the lower-bound
\begin{equation}\label{entrop3}
\frac{H_{n}(v)}{n}=\frac{H_{5k}(v)}{5k}\leq \frac{H_5(v)}{5}\,,
\end{equation}
which is strictly smaller than the single-run min-entropy: $H_{n}(v)/n<H_1(v)$, as illustrated in Figures~\ref{fig:HnTONS} and~\ref{fig:HnABNS}.

\begin{figure}
\includegraphics{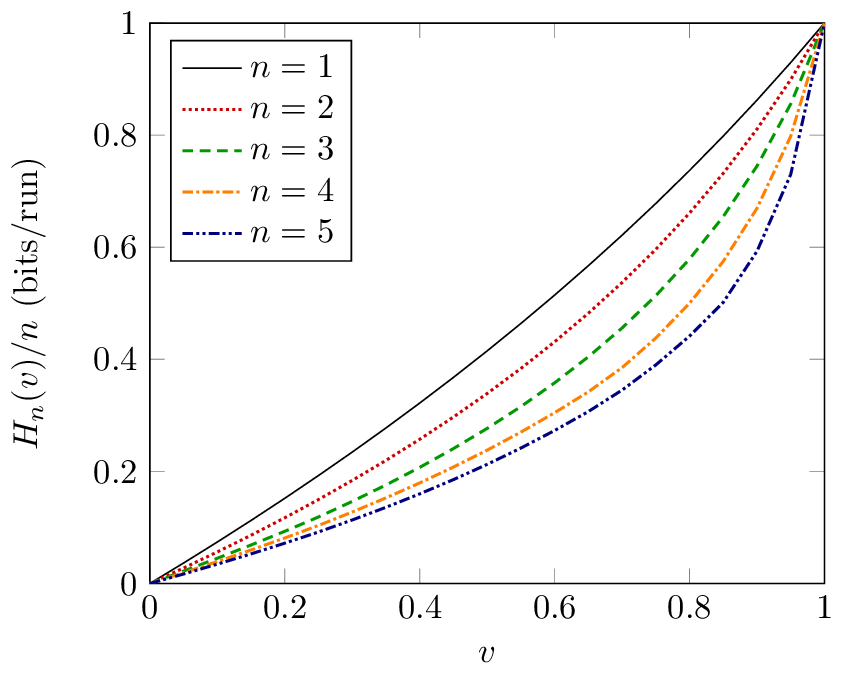}
\caption{Min-entropy rates obtained when Eve is able to guess individual outcomes bits, and pairs, triples, quadruples and quintuples of outcome bits, in the TONS scenario. The curves for the WTONS scenario are virtually the same, as the guessing probabilities for WTONS are either the same or very close to the ones for TONS.}
\label{fig:HnTONS}
\end{figure}

In other words, for multiple uses of the noisy PR-correlations, each instance of the PR-correlations carry less entropy than what one would have naively guessed from (\ref{SolutionN1}). 
This suggest, in analogy with other measures in quantum information, an asymptotic definition $\lim_{n\to\infty} \frac{H_n(v)}{n}$ of the randomness of noisy PR-correlations in the ABNS, TONS, and WTONS scenarios.

\begin{figure}
\includegraphics{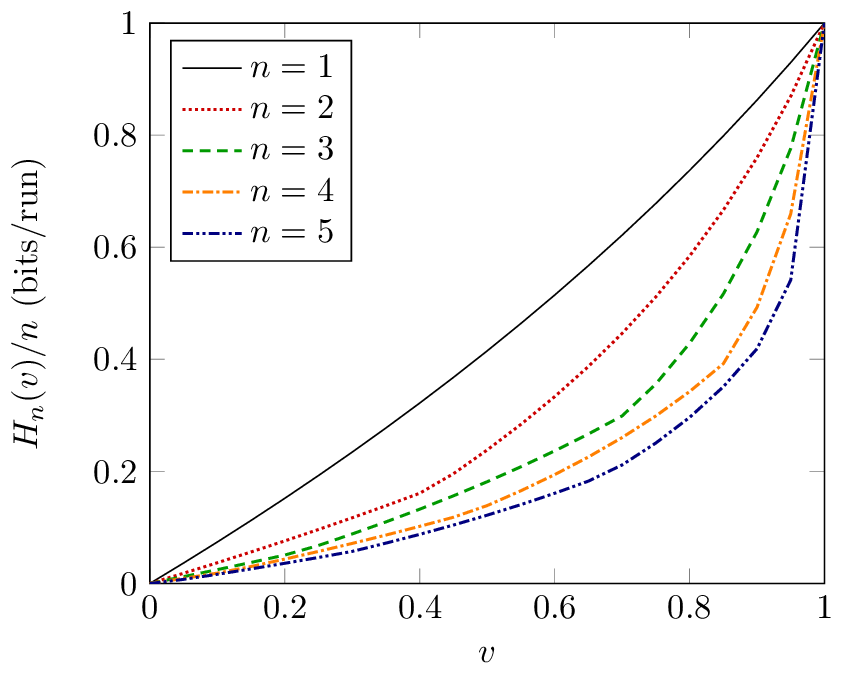}
\caption{Min-entropy rates obtained when Eve is able to guess individual outcomes bits, and pairs, triples, quadruples and quintuples of outcome bits, in the ABNS scenario.}
\label{fig:HnABNS}
\end{figure}

\section{Discussion}
We have investigated the randomness of $n$ noisy PR-boxes, which represent the paradigmatic example of non-local correlations and which are at the basis of many device-independent random number generation and quantum key distribution protocols. 

In the Full-NS scenario, where the $n$ noisy PR-correlations are obtained from $n$ pairs of -- possibly correlated -- but non-communicating devices, the probability to guess correctly the $n$ output bits of one party decreases exponentially with $n$, exactly as if the $n$ noisy PR-boxes were $n$ independent coins with a bias given by Eq. (\ref{SolutionN1}).
 
However, in the ABNS, TONS, and WTONS scenarios, where the $n$ noisy PR-correlations originate from a single pair of devices, either used sequentially $n$ times, or which produce $n$ outcome bits in one run, we have found that the randomness per PR-box can be significantly less than the individual randomness (\ref{SolutionN1}). In particular, we have found that, in the ABNS case, for noise values $v$ below some threshold, the \emph{total randomness} associated to $n\leq 5$ noisy PR-boxes is equal to the randomness of a single noisy PR-box. We conjecture that this holds for any $n$ for some suitable noise threshold. In the TONS and WTONS case, for the same values of $n$, we have found that the guessing probability is linear in $v$ for some region $[v_c^n ; 1]$. We conjecture that this holds for any $n$, but that $v_c^n$ tends to 1.  

Besides their fundamental interest, it is worth considering our results from the perspective of the current status of the security of device-independent random number generation and quantum key distribution protocols. In the Full-NS scenario, their security has been proven \cite{Masanes2014Full}. In the case of the ABNS scenario, there exists a no-go result: starting from $n$ noisy PR-boxes, it is not possible to extract, after privacy amplification, even a single bit that is arbitrarily close to uniform no matter how large $n$ is \cite{Hanggi2013Impossibility} (except if no noise is present, corresponding to $v=1$). In the case of the TONS and WTONS scenarios, the situation is less clear. Though there exist severe limitations on the randomness one can extract from $n$ noisy PR-boxes after privacy amplification \cite{Arnon2012Limits}, those results do not imply that DI RNG or QKD are necessarily impossible in these scenarios.

Interestingly, the Full-NS scenario, where security has been established, corresponds to the situation where the randomness of $n$ noisy PR-boxes accumulates with $n$ in an i.i.d.\ way, while in the ABNS, TONS, and WTONS scenarios, were security was proven to be impossible or is still an open question, the randomness per use of the PR-boxes decreases with $n$. Though the negative results that are presently known for the ABNS, TONS, and WTONS scenarios \cite{Arnon2012Limits} are obtained by taking into account limitations on privacy amplification in a no-signaling context, it is possible that these impossibility results can be traced back to a lack of randomness even before privacy amplification. 

To answer this question definitely, one would have to show that the smooth min-entropy is bounded by a sublinear (i.e. logarithmic or constant) function of $n$. The upper-bounds that we have obtained here are only concerned with the min-entropy, and thus do not imply any such impossibility result. Nevertheless, we believe that they pave the way to a new approach for studying the possibility of no-signalling privacy amplification, as no results were known concerning the min-entropy (smooth or non-smooth) in that context. Though our results do not exclude, in the ABNS, TONS and WTONS scenarios, a linear increase of the min-entropy in the asymptotic limit $n \rightarrow \infty$, they imply an increase at a rate that is significantly lower than what one would naively deduce from the single-copy value (\ref{SolutionN1}). 

\section*{Acknowledgements}

Support from the ERC CoG QITBOX, the AXA Chair in Quantum Information Science, the Spanish MINECO (QIBEQI FIS2016-80773-P and Severo Ochoa SEV-2015-0522), the Generalitat de Catalunya (SGR1381 and CERCA Program), the Fundaci\'o Privada Cellex, and the Fondation Wiener-Anspach is acknowledged. BB acknowledges support from the Secretaria d'Universitats i Recerca del Departament d'Economia i Coneixement de la Generalitat de Catalunya and the European Social Fund - FEDER. SP is a Research Associate of the Fonds de la Recherche Scientifique (F.R.S.-FNRS).

\bibliographystyle{apsrev4-1}
\bibliography{TONSrandomness}

\begin{widetext}
\section*{APPENDICES}
\appendix

In Appendix~\ref{app:OneRound}, we prove the guessing probability value for $n=1$ (Eq.~\ref{SolutionN1}). We then prove two claims made in the main text: $G_n$ is independent of the choice of inputs $(\mathbf{x}^*,\mathbf{y}^*)$ (App.~\ref{app:Symmetries});   the product of $n$ perfect PR-correlations is a vertex of the polytopes associated with any of the no-signaling type constraints introduced in this work (App.~\ref{app:ProductPR}).

We proceed with the proofs of all the guessing probability values given in Table~\ref{tab:1}. We first simplify the optimization problem (\ref{opt-P-general}) using symmetry arguments, and we give its general expression, as well as its associated dual formulation, that follow from these symmetries (App.~\ref{app:PrimalDual}). The detailed expression of the feasible points for these two problems that yield the same objective function value can be found in~\cite{Bourdoncle2018Quantifying}. We provide the necessary information about how to read these files in~Appendix~\ref{app:SolPrimalDual}.

\section{Solution for $n=1$}
\label{app:OneRound}

We first give a feasible point for~\eqref{opt-P-general} that attains the bound given in Eq.~\eqref{SolutionN1}. Let $\{D_i\}_{i=1}^{4}$ be four deterministic behaviors defined as:
\begin{eqnarray}
\begin{aligned}
D_1(a,b|x,y)&= \delta_{a,0}\delta_{b,0}, \\
D_2(a,b|x,y)&= \delta_{a,x} \delta_{b,0}, \\ 
D_3(a,b|x,y)&= \delta_{a,0}\delta_{b,y}, \\
D_4(a,b|x,y)&= \delta_{a,x} \delta_{b,y+1}.
\end{aligned}
\label{DetBox}
\end{eqnarray}
Take 
\begin{align}
& P(\alpha=0)=P(\alpha=1)=\frac{1}{2}, \nonumber \\
& P_{\alpha=0}(ab|xy)=\frac{1-v}{4}(D_1(ab|xy)+D_2(ab|xy)+D_3(ab|xy)+D_4(ab|xy))+v\mathrm{PR}_1(ab|xy), \\[5pt]
& P_{\alpha=1}(ab|xy)=P_{\alpha=0}(\overline{ab}|xy), \nonumber
\end{align}
where, for some $s \in \{0,1\}$,  $\overline{s}$ denotes its complement. Then $\{P(\alpha),P_{\alpha}\}_{\alpha \in \{0,1\}}$ is a feasible point for~\eqref{opt-P-general} that has objective value $1-\frac{v}{2}$. 

Note now that, when $n=1$, Eq.~\eqref{eq:mas} implies $G_1(v) \leq 1-\frac{v}{2}$. This concludes the proof of Eq.~\eqref{SolutionN1}.

\section{Symmetries of the guessing probability problem}
\label{app:Symmetries}

The following transformations allow us to express (\ref{opt-P-general}) in a reduced form.    
\begin{lemma}
Let $(T_1^i), (T_2^i)$ and $(T_3^i)$ be transformations that map a behavior $\PNrounds$ onto another behavior by re-ordering its inputs and outputs in the following way:
\begin{equation*}
(T_1^i) :  \left\{ 
\begin{array}{l} 
a_{i} \rightarrow \overline{a_{i}} \\
b_{i} \rightarrow \overline{b_{i}} \\
\alpha_{i} \rightarrow  \overline{\alpha_{i}} 
\end{array}
\right., \qquad \qquad
 (T_2^i) :  \left\{ 
\begin{array}{l} 
a_{i} \rightarrow a_{i} \oplus x_{i} \\
y_{i} \rightarrow \overline{y_{i}} \\
\end{array}
\right.,  \qquad \qquad
 (T_3^i) :  \left\{ 
\begin{array}{l} 
b_{i} \rightarrow b_{i} \oplus y_{i} \\
x_{i} \rightarrow \overline{x_{i}} \\
\end{array}
\right. .
\end{equation*}

Then, for all i and for all the \textup{NS} conditions, $(T_1^i), (T_2^i)$ and $(T_3^i)$ map a feasible point for \eqref{opt-P-general} onto another feasible point. Moreover,  $(T_1^i)$ preserves the objective function value for all possible  \textup{NS} conditions, and $(T_2^i)$ preserves the objective function value for \textup{\{Full-NS, ABNS, TONS\}}.
\label{lem:Transfo}
\end{lemma}

\begin{proof}
We first prove that a feasible point is mapped onto another feasible point. For a given round $i$, let $(T_j^i)$ be one these transformations and let $\{P(\boldsymbol{\alpha}), \PNrounds \}$ be a feasible point for~\eqref{opt-P-general} for some NS condition. Let $\{\tilde{P}(\boldsymbol{\alpha}), \tilde{P}_{\boldsymbol{\alpha}}(\mathbf{a},\mathbf{b}|\mathbf{x},\mathbf{y}) \}$ be the image of this point by $(T_j^i)$. Since the NS condition involves all $(\mathbf{a,b,x,y})$, and since $(T_j^i)$ simply reorders some elements of $\{P(\boldsymbol{\alpha}), \PNrounds \}$ in an individual round $i$,  $\{\tilde{P}(\boldsymbol{\alpha}), \tilde{P}_{\boldsymbol{\alpha}}(\mathbf{a},\mathbf{b}|\mathbf{x},\mathbf{y}) \}$ satisfies the same NS condition. Moreover, since the behavior $\mathrm{PR}_v(a_ib_i|x_iy_i)$ is invariant under $(T_j^i)$, $\{\tilde{P}(\boldsymbol{\alpha}), \tilde{P}_{\boldsymbol{\alpha}}(\mathbf{a},\mathbf{b}|\mathbf{x},\mathbf{y}) \}$ also satisfies the constraint on the marginals. $(T_j^i)$  thus maps a feasible point for (\ref{opt-P-general}) onto another feasible point.

We now show that $(T_1^i)$ preserves the objective function value of~\eqref{opt-P-general} for all the NS conditions. For simplicity, let us take $i=1$, the argument for $i>1$ being the same. Let $\{\tilde{P}(\boldsymbol{\alpha}), \tilde{P}_{\boldsymbol{\alpha}}(\mathbf{a},\mathbf{b}|\mathbf{x},\mathbf{y}) \}$ be the image of a feasible point $\{P(\boldsymbol{\alpha}), \PNrounds \}$ by $(T_1^1)$. Then: 
\begin{equation}
\sum_{\boldsymbol{\alpha},\mathbf{b}} \tilde{P}(\boldsymbol{\alpha})\tilde{P}_{\boldsymbol{\alpha}}(\boldsymbol{\alpha},\mathbf{b}|\mathbf{0,0}) =  \sum_{\boldsymbol{\alpha},\mathbf{b}} P(\overline{\alpha_1}\boldsymbol{\alpha_{>1}}){P}_{\overline{\alpha_1}\boldsymbol{\alpha_{>1}}}(\overline{\alpha_1}\boldsymbol{\alpha_{>1}},\overline{b_1}\mathbf{b_{>1}}|\mathbf{0,0}) = \sum_{\boldsymbol{\alpha},\mathbf{b}} P(\boldsymbol{\alpha})P_{\boldsymbol{\alpha}}(\boldsymbol{\alpha},\mathbf{b}|\mathbf{0,0}) 
\end{equation}

We now show that $(T_2^i)$ preserves the objective function value of~\eqref{opt-P-general} for all but the WTONS condition. We again set $i=1$, and denote $\{\tilde{P}(\boldsymbol{\alpha}), \tilde{P}_{\boldsymbol{\alpha}}(\mathbf{a},\mathbf{b}|\mathbf{x},\mathbf{y}) \}$ the image by $(T_2^1)$ of a feasible point $\{P(\boldsymbol{\alpha}), \PNrounds \}$ for the Full-NS, ABNS or TONS condition. Then: 
\begin{equation}
\begin{aligned}
\sum_{\boldsymbol{\alpha},\mathbf{b}} \tilde{P}(\boldsymbol{\alpha})\tilde{P}_{\boldsymbol{\alpha}}(\boldsymbol{\alpha},\mathbf{b}|\mathbf{0,0}) & =  \sum_{\boldsymbol{\alpha},\mathbf{b}} P(\boldsymbol{\alpha}){P}_{\boldsymbol{\alpha}}(\boldsymbol{\alpha},\mathbf{b}|\mathbf{0},10\hdots0) \\
& = \sum_{\boldsymbol{\alpha}} P(\boldsymbol{\alpha}) \sum_{\mathbf{b}}P_{\boldsymbol{\alpha}}(\boldsymbol{\alpha},\mathbf{b}|\mathbf{0},10\hdots0) \\
& =  \sum_{\boldsymbol{\alpha}} P(\boldsymbol{\alpha}) \sum_{\mathbf{b}}P_{\boldsymbol{\alpha}}(\boldsymbol{\alpha},\mathbf{b}|\mathbf{0},00\hdots0) 
\end{aligned}
\end{equation}
where the last equality holds because, for all $\alpha$, $P_{\boldsymbol{\alpha}}$ is ABNS.
\end{proof}

Thanks to $(T_2^i)$ and $(T_3^i)$, we can now prove that the optimal value $G_n(v)$ defined in~\eqref{opt-P-general} is independent of $(\mathbf{x^*,y^*})$. Let us momentarily call $G_n^{\mathbf{x^*,y^*}}(v)$ the solution of~\eqref{opt-P-general}. Let us assume that $\{P(\boldsymbol{\alpha}), \PNrounds \}$ is a feasible point for~\eqref{opt-P-general} that achieves the value $G_n^{\mathbf{0,0}}(v)$. We then construct $\{P(\boldsymbol{\alpha}), \tilde{P}_{\boldsymbol{\alpha}}(\mathbf{a},\mathbf{b}|\mathbf{x},\mathbf{y}) \}$ by applying $(T_2^1)$ onto $\{P(\boldsymbol{\alpha}), \PNrounds \}$. Then: 
\begin{equation}
\begin{aligned}
\sum_{\boldsymbol{\alpha},\mathbf{b}} P(\boldsymbol{\alpha})\tilde{P}_{\boldsymbol{\alpha}}(\boldsymbol{\alpha},\mathbf{b}|00 \hdots 0,10 \hdots 0) & =  \sum_{\boldsymbol{\alpha},\mathbf{b}} P(\boldsymbol{\alpha}){P}_{\boldsymbol{\alpha}}(\boldsymbol{\alpha},\mathbf{b}|00 \hdots 0,00 \hdots 0)  \\
& = G_n^{\mathbf{00}}(v)
\end{aligned}
\end{equation}
This implies $G_n^{00 \hdots 0,10 \hdots 0}(v) \geq G_n^{\mathbf{0,0}}(v)$. Let us now assume that $\{P(\boldsymbol{\alpha}), \PNrounds \}$ is a feasible point for~\eqref{opt-P-general} that achieves the value $G_n^{00 \hdots 0,10 \hdots 0}(v)$, and construct $\{P(\boldsymbol{\alpha}), \tilde{P}_{\boldsymbol{\alpha}}(\mathbf{a},\mathbf{b}|\mathbf{x},\mathbf{y}) \}$ by applying $(T_2^1)$ onto it. Then: 
\begin{equation}
\begin{aligned}
\sum_{\boldsymbol{\alpha},\mathbf{b}} P(\boldsymbol{\alpha})\tilde{P}_{\boldsymbol{\alpha}}(\boldsymbol{\alpha},\mathbf{b}|00 \hdots 0,00 \hdots 0) & =  \sum_{\boldsymbol{\alpha},\mathbf{b}} P(\boldsymbol{\alpha}){P}_{\boldsymbol{\alpha}}(\boldsymbol{\alpha},\mathbf{b}|00 \hdots 0,10 \hdots 0)  \\
& = G_n^{00 \hdots 0,10 \hdots 0}(v)
\end{aligned}
\end{equation}
This implies $G_n^{\mathbf{0,0}}(v) \geq G_n^{00 \hdots 0,10 \hdots 0}(v)$ and thence $G_n^{\mathbf{00}}(v) = G_n^{00 \hdots 0,10 \hdots 0}(v)$. The same construction can be done for all other values of $\mathbf{y}$ by applying $(T_2^i)$ whenever $y_i = 1$, thus proving $G_n^{\mathbf{0,0}}(v) = G_n^{\mathbf{0,y}}(v)$ for all $\mathbf{y}$.

We now assume that $\{P(\boldsymbol{\alpha}), \PNrounds \}$ is a feasible point for~\eqref{opt-P-general} that achieves the value $G_n^{\mathbf{0}, \mathbf{y}}(v)$. For some $\mathbf{x} \in \{0,1\}^n$, we construct $\{P(\boldsymbol{\alpha}), \tilde{P}_{\boldsymbol{\alpha}}(\mathbf{a},\mathbf{b}|\mathbf{x},\mathbf{y}) \}$ by applying $(T_3^i)$ onto it whenever $x_i=1$. Then: 
\begin{equation}
\begin{aligned}
\sum_{\boldsymbol{\alpha},\mathbf{b}} P(\boldsymbol{\alpha})\tilde{P}_{\boldsymbol{\alpha}}(\boldsymbol{\alpha},\mathbf{b}|\mathbf{x,y}) =  \sum_{\boldsymbol{\alpha},\mathbf{b}} P(\boldsymbol{\alpha}){P}_{\boldsymbol{\alpha}}(\boldsymbol{\alpha},\mathbf{b}\oplus \mathbf{xy}|\mathbf{0,y}) & =  \sum_{\boldsymbol{\alpha},\mathbf{b}} P(\boldsymbol{\alpha}){P}_{\boldsymbol{\alpha}}(\boldsymbol{\alpha},\mathbf{b}|\mathbf{0,y})  \\
& = G_n^{\mathbf{0,y}}(v)
\end{aligned}
\end{equation}
where the first equality holds because we applied $(T_3^i)$ only when $x_i=1$ and the second one holds because we sum over $\mathbf{b}$. This implies that $G_n^{\mathbf{x,y}}(v) \geq G_n^{\mathbf{0,y}}(v)$. Let us now assume that $\{P(\boldsymbol{\alpha}), \PNrounds \}$ is a feasible point for~\eqref{opt-P-general} that achieves the value $G_n^{\mathbf{x}, \mathbf{y}}(v)$. We construct $\{P(\boldsymbol{\alpha}), \tilde{P}_{\boldsymbol{\alpha}}(\mathbf{a},\mathbf{b}|\mathbf{x},\mathbf{y}) \}$ by applying $(T_3^i)$ onto it whenever $x_i=1$. Then:
\begin{equation}
\begin{aligned}
\sum_{\boldsymbol{\alpha},\mathbf{b}} P(\boldsymbol{\alpha})\tilde{P}_{\boldsymbol{\alpha}}(\boldsymbol{\alpha},\mathbf{b}|\mathbf{0,y}) =  \sum_{\boldsymbol{\alpha},\mathbf{b}} P(\boldsymbol{\alpha}){P}_{\boldsymbol{\alpha}}(\boldsymbol{\alpha},\mathbf{b}\oplus \mathbf{xy}|\mathbf{x,y}) & =  \sum_{\boldsymbol{\alpha},\mathbf{b}} P(\boldsymbol{\alpha}){P}_{\boldsymbol{\alpha}}(\boldsymbol{\alpha},\mathbf{b}|\mathbf{x,y})  \\
& = G_n^{\mathbf{x,y}}(v)
\end{aligned}
\end{equation}
This implies $G_n^{\mathbf{0,y}}(v) \geq G_n^{\mathbf{x,y}}(v)$ and thence $G_n^{\mathbf{x,y}}(v) = G_n^{\mathbf{0,y}}(v)$. 

Altogether, this proves that $G_n^{\mathbf{x,y}}(v) = G_n^{\mathbf{0,0}}(v)$ for all $(\mathbf{x,y})$, and thus that $G_n(v)$ is properly defined.

\section{Product of $n$ perfect PR-correlations}
\label{app:ProductPR} 

We now show that the product of $n$ PR-boxes is a vertex of any of the no-signaling polytopes we introduced in the main text. We do it for $n=2$, the generalization to $n \geq 3$ is straightforward. Let us assume that there exists two ABNS (resp. WTONS) joint distributions $P_1$ and $P_2$ such that:
\begin{equation}
\mathrm{PR}_1(a_1,b_1|x_1,y_1) \times \mathrm{PR}_1(a_2,b_2|x_2,y_2) = \lambda_1 P_1(\mathbf{a,b}|\mathbf{x,y})+\lambda_2 P_2(\mathbf{a,b}|\mathbf{x,y})
\end{equation} 
for some $(\lambda_1,\lambda_2)\in [0,1]$ such that $\lambda_1 + \lambda_2=1$. 

Then:
\begin{equation}
\begin{aligned}
\mathrm{PR}_1(a_1,b_1|x_1,y_1) & =  \lambda_1 \sum_{a_2,b_2} P_1(\mathbf{a,b}|\mathbf{x,y})+\lambda_2  \sum_{a_2,b_2} P_2(\mathbf{a,b}|\mathbf{x,y}) \\
& =  \lambda_1 P_1(a_1,b_1|x_1,x_2,y_1,y_2) + \lambda_2 P_2(a_1,b_1|x_1,x_2,y_1,y_2).
\end{aligned}
\label{ConvexDecomposition}
\end{equation}

Let us fix a specific value $(x_2^*,y_2^*)$ for $(x_2,y_2)$. Then $P_1(a_1,b_1|x_1,x_2^*,y_1,y_2^*)$ is a no-signaling bipartite binary behavior. Indeed,
\begin{equation}
\sum_{b_1} P_1(a_1,b_1|x_1,x_2^*,y_1,y_2^*) = \sum_{a_2, b_1, b_2} P_1(a_1, a_2,b_1, b_2|x_1,x_2^*,y_1,y_2^*)
\end{equation}
is independent of $y_1$ because $P_1$ is ABNS (resp. WTONS), and $\sum_{a_1} P_1(a_1,b_1|x_1,x_2^*,y_1,y_2^*)$ is independent of $x_1$ for the same reason. The same goes for $P_2(a_1,b_1|x_1,x_2^*,y_1,y_2^*)$. Since the PR-box is a vertex of the polytope of bipartite binary no-signaling behaviors, Eq.~\eqref{ConvexDecomposition} implies 
\begin{equation}
P_1(a_1,b_1|x_1,x_2^*,y_1,y_2^*) = P_2(a_1,b_1|x_1,x_2^*,y_1,y_2^*) = \mathrm{PR}_1(a_1,b_1|x_1,y_1) 
\end{equation} 
for all values of $(x_2^*,y_2^*)$. The same holds for $P_1(a_2,b_2|x_1^*,x_2,y_1^*,y_2)$ and  $P_2(a_2,b_2|x_1^*,x_2,y_1^*,y_2)$, for all values of $(x_1^*,y_1^*)$. This implies: 
\begin{equation}
P_1(\mathbf{a,b}|\mathbf{x,y})=P_2(\mathbf{a,b}|\mathbf{x,y})=\mathrm{PR}_1(a_1,b_1|x_1,y_1) \times \mathrm{PR}_1(a_2,b_2|x_2,y_2).
\end{equation} 

The product of two PR-boxes cannot be decomposed over different joint distributions in ABNS (resp. WTONS): it is thus a vertex of the ABNS (resp. WTONS) polytope. Since Full-NS and TONS are subsets of these polytopes, it also implies that it is a vertex of Full-NS and TONS.  

\section{Primal and dual form of the guessing probability problem} 
\label{app:PrimalDual}

The symmetry $(T_1^i)$ given in Appendix~\ref{app:Symmetries} implies that the solutions to the problem defined by Equation (\ref{opt-P-general}) can be found in the reduced space:
\begin{equation}
\mathcal{S}=\Big\{ (P(\boldsymbol{\alpha}),\PNrounds) \Big\rvert \PNrounds = P_{\boldsymbol{\alpha}=\mathbf{0}} (\mathbf{\overline{a,b}^{\alpha}}|\mathbf{x,y}), P(\boldsymbol{\alpha}) = \frac{1}{2^n}\Big\}
\label{eq:ReducedSpace}
\end{equation}
where $\overline{a_i,b_i}^{\alpha_i}=\begin{cases} a_i,b_i \textrm{ if } \alpha_i=0, \\ \overline{a_i},\overline{b_i} \textrm{ if } \alpha_i=1. \end{cases}$

From here on, we'll thus only consider distributions with such symmetries, and we'll write $P(\mathbf{a,b}|\mathbf{x,y})$ for $P_{\boldsymbol{\alpha}=\mathbf{0}}(\mathbf{a,b}|\mathbf{x,y})$. Note that, for $P \in \mathcal{S}$, the objective function of (\ref{opt-P-general}) becomes 
\begin{equation}
\sum_{\boldsymbol{\alpha}, \mathbf{b}}P(\boldsymbol{\alpha})P_{\boldsymbol{\alpha}}(\boldsymbol{\alpha},\mathbf{b}|\mathbf{0,0})=\sum_{\boldsymbol{\alpha}, \mathbf{b}}\frac{1}{2^n}P_{\boldsymbol{0}}(\mathbf{0,b}|\mathbf{0,0})=\sum_{\mathbf{b}}P(\mathbf{0,b}|\mathbf{0,0}).
\end{equation}

Moreover, a constraint on the marginals is now expressed in the following way:
\begin{equation}
\begin{aligned}
\sum_{\boldsymbol{\alpha}} P(\boldsymbol{\alpha})\PNrounds=\prod_{i=1}^n \mathrm{PR}_v(a_i,b_i|x_i,y_i)  & \Leftrightarrow \frac{1}{2^n} \sum_{\boldsymbol{\alpha}} P_{\boldsymbol{\alpha}=\mathbf{0}} (\mathbf{\overline{a,b}^{\alpha}}|\mathbf{x,y}) = \prod_{i=1}^n \mathrm{PR}_v(a_i,b_i|x_i,y_i) \\
& \Leftrightarrow \frac{1}{2^n} \sum_{\boldsymbol{a}} P (\mathbf{a,a\oplus b}|\mathbf{x,y}) = \prod_{i=1}^n \mathrm{PR}_v(0,b_i|x_i,y_i). 
\end{aligned}
\end{equation}

The optimization problem defined in (\ref{opt-P-general}) can thus be written as:
\begin{equation}
\begin{array}{>{\displaystyle}r<{\displaystyle}>{\displaystyle}c<{\displaystyle}>{\displaystyle}l<{\displaystyle}}
G_n(v) = & \max & \sum_{\mathbf{b}}P(\mathbf{0,b}|\mathbf{0,0}) \\[10pt]
& \mathrm{s.t.} & \forall (\mathbf{x}, \mathbf{y}, \mathbf{b}) \in \{0,1\}^{3n}, \ \sum_{\boldsymbol{a}} P (\mathbf{a,a\oplus b}|\mathbf{x,y}) = 2^n \times \prod_{i=1}^n \mathrm{PR}_v(0,b_i|x_i,y_i)\\[12pt]
& & P \in \mathrm{NS}
\end{array}
\label{New-opt-P}
\end{equation}

In order to construct the dual of~\eqref{New-opt-P}, note that $P$ can be seen as a vector, on which two kinds of constraints apply: on the one hand positivity, as it represents some probability distributions, on the other hand linear constraints, that arise both from the marginal constraints and the no-signaling scenario that is considered. 

The optimization problem~\eqref{New-opt-P} and its associated dual problem can then be summarized as:

\noindent
\begin{minipage}{0.45\linewidth}
\begin{equation}
\begin{array}{>{\displaystyle}r<{\displaystyle}>{\displaystyle}c<{\displaystyle}>{\displaystyle}l<{\displaystyle}}
G_n(v) = & \max & \mathbf{c}^\top \mathbf{p} \\[10pt]
& \mathrm{s.t.} & A\mathbf{p}=\mathbf{b}\\[10pt]
& & \mathbf{p} \geq 0
\end{array}
\label{Primal}
\end{equation}
\end{minipage}
\begin{minipage}{0.45\linewidth}
\begin{equation}
\begin{array}{>{\displaystyle}r<{\displaystyle}>{\displaystyle}c<{\displaystyle}>{\displaystyle}l<{\displaystyle}}
G_n(v) = & \min & \mathbf{b}^\top \mathbf{y} \\[10pt]
& \mathrm{s.t.} & A^\top \mathbf{y} \geq \mathbf{c}
\end{array}
\label{Dual}
\end{equation}
\end{minipage}

\vspace{10pt}
\noindent
where and $A$ and $\mathbf{b}$ describe the marginal and no-signaling constraints and $c_i=\begin{cases} 1 \textrm{ if } p_i = P(\mathbf{0},\mathbf{b} | \mathbf{0,0}), \\ 0 \textrm{ otherwise. } \end{cases}$

Strong duality holds here because~\eqref{Primal} is linear and feasible (the target correlation, i.e., $n$ noisy i.i.d. PR boxes, is always a solution). This implies that finding the optimum now amounts to finding feasible points for these two problems that yield the same objective function value.

\section{Solutions of the primal and dual problems}
\label{app:SolPrimalDual}

The solutions of \eqref{Primal} and \eqref{Dual} when $n=2,3$ can be found in~\cite{Bourdoncle2018Quantifying}. Since $G^n(v)$ is the same for TONS and WTONS when $n=2,3$, we give only a primal feasible point for TONS and a dual feasible point for WTONS with the same objective function value, which is sufficient to prove the values given in Table~\ref{tab:1} for TONS and WTONS. Indeed, let us call momentarily $p_{TONS}^*$ (resp. $p_{WTONS}^*$) the solution of~\eqref{Primal} for TONS (resp. WTONS), and $d_{WTONS}^*$ the solution of~\eqref{Dual} for WTONS. Let us call $p_{TONS}$ the objective function value associated to our primal feasible point for TONS, and $d_{WTONS}$ the objective function value associated to our dual feasible point for WTONS. We then have 
\begin{align}
p_{TONS} & \leq p_{TONS}^*,\\
d_{WTONS}^* & \leq d_{WTONS}.
\end{align}

Moreover, $p_{TONS}^* \leq p_{WTONS}^*$ because TONS $\subset$ WTONS and $p_{WTONS}^* = d_{WTONS}^*$ because strong quality holds. Altogether, this gives:
\begin{equation}
p_{TONS} \leq p_{TONS}^* \leq d_{WTONS}^* \leq d_{WTONS}.
\end{equation}
Finding a primal feasible point for TONS and a dual feasible point for WTONS such that $p_{TONS}=d_{WTONS}$ is thus sufficient to solve~\eqref{opt-P-general} both for TONS and WTONS. 

For the solutions of~\eqref{Primal}, we give only $P(\mathbf{a,b}|\mathbf{x,y=0})$: since the symmetry $(T_2^i)$ is valid both for TONS and ABNS, the distributions for other values of $\mathbf{y}$ can be derived from $P(\mathbf{a,b}|\mathbf{x,y=0})$ alone, by applying the corresponding transformation. 

For the solutions of~\eqref{Dual} to be defined without ambiguity, the order of the constraints listed in the matrix $A$ and vector $\mathbf{b}$ should be fixed. We thus include in~\cite{Bourdoncle2018Quantifying} the scripts that construct the specific matrices $A$ and vectors $\mathbf{b}$ for which our dual solutions are defined. 

\end{widetext}

\end{document}